\documentclass[a4paper, 11pt]{article}
\usepackage{fullpage}
\usepackage{authblk}
\usepackage{graphicx}
\usepackage{nccmath}
\usepackage[hang, small, bf]{caption}
\usepackage{color}
\usepackage[usenames,dvipsnames,svgnames,table]{xcolor}
\usepackage[colorlinks=true,
            linkcolor=blue,
            urlcolor=yellow,
            citecolor=blue]{hyperref}

\usepackage{hyperref}
\usepackage{tikz}
\usepackage{lmodern}
\usepackage{amssymb,amsmath,amsfonts,amsxtra,amsthm}
\usepackage{pstricks,pst-node,pst-text,pst-3d}
\usepackage{wasysym}
\usetikzlibrary{arrows,positioning,automata,shadows,fit,shapes}
\usepackage[utf8]{inputenc}                             
\usepackage[T1]{fontenc}
\usepackage[all,cmtip]{xy}
\usepackage{tikz-cd}
\usepackage[round]{natbib}
\usepackage[british,UKenglish,USenglish,english,american]{}
\usepackage{amsmath,amsthm,amssymb,amsfonts}

\newcommand{\argmax}[2]{%
\smash{\mathop{{\rm argmax}}\limits_{#1}}\, #2 }

\theoremstyle{definition}

\newtheorem{Theorem}{Theorem}
\newtheorem{Proposition}{Proposition}
\newtheorem{Lemma}{Lemma}

\newtheorem{Remark}{Remark}

\usepackage{multirow,array}
\usepackage[textwidth=30mm]{todonotes}

\usepackage{soul}
\usepackage[utf8]{inputenc}
\usepackage[english]{babel}
\setstcolor{red}
\sethlcolor{SkyBlue}

\begin{document}

\title{A Story of Consistency:\\ Bridging the Gap between Bentham and Rawls Foundations}

\author[1]{St\'ephane Gonzalez }
\affil[1]{Univ Lyon, UJM Saint-Etienne, GATE UMR 5824, F-42023 Saint- Etienne, France, \texttt stephane.gonzalez@univ-st-etienne.fr,}
\author[2]{Nikolaos Pnevmatikos}
\affil[2]{LEMMA, Universit\'e Paris 2, Panth\'eon-Assas, 4 rue Blaise Desgoffe, Paris, France, nikolaos.pnevmatikos@u-paris2.fr}

\date{Version of \today}
\maketitle

\begin{abstract}
The axiomatic foundations of Bentham and Rawls solutions are discussed within the broader domain of cardinal preferences. It is unveiled that both solution concepts share all four of the following axioms: \textit{Nonemptiness}, \textit{Anonymity}, \textit{Unanimity}, and \textit{Continuity}. In order to fully characterize the Bentham and Rawls solutions, three variations of a \textit{consistency criterion} are introduced and their compatibility with the other axioms is assessed. Each expression of consistency can be interpreted as a property of decision-making in uncertain environments.
\end{abstract}

\textbf{Keywords}: utilitarianism, egalitarianism, axiomatization, consistency

\

\textbf{JEL Classification}: D71

\vspace{0.5cm}

\section{Introduction} \label{intro}

\hspace{0.4cm} The design of methods for comparing and evaluating the various alternatives available to a given society has been widely studied in social choice theory. A social welfare criterion is often used to measure the utility or benefit of each alternative for the society as a whole. Two criteria, drawn from moral philosophy, have been particularly influential: \cite{bentham1970introduction} utilitarian criterion (henceforth Bentham solution), which considers the utilities of all individuals and chooses the alternative that maximizes the sum of these utilities, and \cite{rawls1971atheory} egalitarian criterion (henceforth Rawls solution), which defends the alternative that provides the greatest benefit to the least well-off member of the society. \\
There is a sharp contrast between these two concepts when it comes to their attitude toward equitable distribution. While utilitarianism is neutral in respect of utility inequality, the egalitarian principle is extremely inequality-averse by paying attention to the worst-off individual only. These two approaches gave rise to a vast literature on social choice, including axiomatic characterizations which attempt to justify their value on the basis of desirable normative principles. \cite{harsanyi1955cardinal}'s aggregation result on utilitarianism constitutes a milestone in this regard. As a result, a significant branch of the literature has been generated, focusing on Bentham solution's axiomatics, in particular with regard to the collective decision making under uncertainty and the traditional economic distributional issues. Along this line, notable examples include \cite{d1977equity}, \cite{maskin1978theorem}, \cite{blackorby1982ratio}, \cite{mongin1994harsanyi}, \cite{weymark1990reconsideration} and \cite{weymark2005measurement}. On the contrary, comparatively few works in the literature provide axiomatic characterisations of the Rawls solution; see, for instance, \cite{hammond1979straightforward}, \cite{d1977equity}, \cite{gevers1979interpersonal}, \cite{ou2018generalized} and \cite{mongin2021rawls}. The reader may consult \cite{d2002social} or \cite{kamaga2018utilitarianism} for an extensive survey on the topic of Bentham and Rawls solutions' axiomatics. 

A first attempt to compare or reconcile the views of Bentham and Rawls, with an aim to recognize the axioms that characterize both solutions simultaneously can be found in \cite{d1977equity} where the authors show that, in the context of social welfare functionals, utilitarianism and the lexical version of the Rawlsian criterion share the properties of \textit{Independence of Irrelevant Alternatives}, \textit{Anonymity} and \textit{Strong Pareto} and there is a single axiom differentiating them; on the one hand, utilitarianism needs the \textit{Individual Origin of Utilities} and on the other, egalitarianism requires the \textit{Extreme Equity} axiom. Similar results are known to hold in a welfarist context where characterisations of the two solutions can be derived based on \textit{Strong Pareto}, \textit{Separability}, \textit{Anonymity} and a suitable informational invariance axiom; \cite{d1994welfarism} states that a theorem can be proved (see the proofs of \cite{d1985axioms} or \cite{d1977equity}) using other characterizations of leximin given by \cite{hammond1976equity} and \cite{strasnick1976social} showing that, ultimately, the only distinguishing property of leximin with respect to classical utilitarianism is the sort of information and the type of interpersonal comparisons it relies on, namely the replacing of \textit{Ordinal non-Comparability} by \textit{Co-Ordinality}, while utilitarianism requires \textit{Cardinality and Unit Comparability}.
\\
Thereafter, numerous articles provide axiomatizations of the two solutions along this path, such as \cite{sen1979utilitarianism}, \cite{deschamps1978leximin}, \cite{myerson1981utilitarianism}, or with the intention of identifying the conditions which ensure that these solutions intersect, such as \cite{kamaga2018utilitarianism}. We further mention \cite{bossert2020axiomatization} who study convex combinations of the two principles with a view to justify them axiomatically and \cite{gajdos2008ignorant} who characterize a convex combination of the utilitarian and egalitarian criteria in which individuals' utilities are cardinally measurable and fully comparable\footnote{The normalization of utility functions the authors obtain is very similar to that appeared in \cite{dhillon1999relative}, \cite{segal2000let} and \cite{moreno2008veil}.}. Earlier, \cite{roemer1996theories} and \cite{roemer1998equality} formalize a mechanism to construct equal-opportunity policies, which could be seen as a compromise between utilitarian and egalitarian criteria. More recently, \cite{silvestre2014utilitarianism} provides an analysis of consumers with either quasi-linear or Gorman-form preferences and shows that under some technical conditions the utilitarian maximum implies equal utilities.

Restrictions on the preferences of agents can be a topic of debate, with different perspectives on the advantages and disadvantages of each approach. One such debate is whether to use ordinal preferences or utility functions to represent the agents' preferences. 
\\
On the one hand, ordinal preferences are often seen as a lean structure bearing in mind that they only require agents to be able to rank alternatives. This simplicity leads to a better understanding and interpretation of the results derived from models in this setting. Moreover, ordinal preferences have been broadly used in the literature, and there is a wealth of existing results and methods that can be applied in this frame. On the other hand, utility functions have the privilege of being able to express the intensity of preferences for each agent. This additional information brings out a more detailed and nuanced comprehension of how agents make decisions. Prominent papers relevant to the latter direction include \cite{maskin1978theorem}, \cite{kamaga2018utilitarianism}, \cite{blackorby1982ratio}. Cardinal preferences allow us to come by similar results to the ones obtained within the ordinal preferences, which are nevertheless mathematically easier to establish. Making use of utility functions to express preferences permits to explore decision-making problems in a more thorough manner, ultimately leading to deeper insights of how agents make decisions.
\\
The purpose of this article is to provide some novel axiomatic characterizations of the Bentham optima and Rawls max-min solutions within the context of utility profiles, following thus an approach that differs from previous ones - which generally rely on ordinal preferences profiles - to the extent that our setting is simple and only requires a basic knowledge of utility function profiles and elementary mathematical tools.
\\
Our intention is to put the spotlight on two solution concepts, which are typically seen as being in conflict, yet they are alike in view of their axiomatizations. Namely, we show that the Bentham and Rawls solutions have four axioms in common, which are well-known in the literature of social choice theory: \textsf{Nonemptiness}, \textsf{Anonymity}, \textsf{Unanimity}, and \textsf{Continuity}. \textsf{Nonemptiness} ensures that at least one alternative is selected as the solution. \textsf{Anonymity} guarantees that all agents are treated equally. \textsf{Unanimity} makes sure that if an alternative is preferred by every agent, then it is selected as a solution, which corresponds to a weaken version of Pareto optimality concept. \textsf{Continuity} requires that the solution is stable with respect to small perturbations in the utility profile.
\\
We introduce a consistency principle such that if two distinct utility profiles share a common alternative as a solution, then this alternative should also be a solution in the aggregated utility profile. This rule is akin to the \textit{mixture preserving property} introduced by \cite{harsanyi1955cardinal} and widely used in the literature which characterizes the \cite{vontheory} utility functions. It requires the aggregated utility of a convex combination of inputs to be equal with the convex combination of the outputs, which is a fundamental property of the mathematical framework of cardinal utility. Our principle is closely related to the \textit{consistency property} recently studied by \cite{brandl2019justifying} in the setting of two-player zero-sum games where consistency is established by requiring that if both players choose the same strategy pair in two given games, then this pair should also be applied in a new game where the payoffs are determined by the randomization of the payoffs of the two primal games.\\
Throughout this paper we consider agents dealing with a situation where their utility profile is conditional on a random event. We introduce three variations of consistency according to the type of beliefs whereby the agents think that each utility profile may occur.
\begin{enumerate}
    \item The first variation is called \textsf{Subjective Expected Consistency}. Whenever an alternative is selected in two different utility profiles, it indicates that this alternative should also be selected in the expected utility profile, regardless of possibly heterogeneous agents' beliefs on the occurrence of each profile. It seems to be compelling when personal beliefs or opinions about the likelihood of a particular event occurring vary from agent to agent and are influenced by agent's experiences, knowledge, and preferences. Precisely, given two utility profiles and a probability $p_i$ over them for each agent $i \in N$, the solution of the expected utility profile under $\textbf{p}=(p_1,...,p_n)$ must match the intersection of the solution sets of the two utility profiles, unless it is empty. The \textsf{Subjective Expected Consistency} principle is founded on the notion of subjective probability, that is to say, a type of probability derived from an agent's personal judgment or own experience about whether a specific event is likely to occur.
      \item The second variation is called \textsf{Objective Expected Consistency}. It restricts the idea of the former variation to the case of objective probabilities. It is particularly relevant to the situation wherein the utility profile is determined by a random event such as a coin toss. It points out that an alternative which is selected in both utility profiles should be also selected before the coin toss. In formal terms, given two utility profiles and a probability $p$ over them, the solution set of the expected utility profile under $p$ must match the intersection of the solution sets of the two utility profiles, unless it is empty. The \textsf{Objective Expected Consistency} principle is based on the concept of objective probability, simply put, a probability which is related to objective observable facts and independent of any individual's beliefs or preferences.    
    \item The third variation of the consistency principle, referred to as \textsf{Minimum Consistency}, replaces the expected utility profile by the worst-case scenario profile. It stipulates that an alternative which is selected in both utility profiles, should also be selected in the worst-case scenario between the two utility profiles. In the worst-case scenario profile, at any given alternative, each agent's utility is defined as the lowest among the two utility profiles. In  particular, this variation is useful when individuals have no beliefs on the probability of the random event, since it allows to make decision-making in uncertain situations which are not relying on subjective expectations.
\end{enumerate}

\

\subsection*{Illustrative example}
To illustrate our result, consider a situation where two agents, Jeremy and John, collaborate on a research project. In doing so, they meet in Oxford to work together. They can either work at the laboratory of the university ($\ell$) or at the caf\'e of their neighborhood ($c$). The preferences of the two agents regarding the working place depend on the weather and are represented by their utilities over the two possible alternatives. When the weather is sunny, the utilities of the two agents are given by: 
\begin{center}
$\textbf{u}=\begin{bmatrix}
   u_1(\ell)  & u_1(c) \\
   u_2(\ell)  & u_2(c)
\end{bmatrix}$
\end{center}
and when it is raining their utilities are given by:
\begin{center}
$\textbf{v}=\begin{bmatrix}
   v_1(\ell)  & v_1(c) \\
   v_2(\ell)  & v_2(c)
\end{bmatrix}$
\end{center}

In order to make their choices, Jeremy and John need to design a decision process $\varphi$ which associates to each utility profile $\textbf{u}$ an alternative $\varphi(\textbf{u})$, i.e., laboratory ($\ell$) or caf\'e ($c$).\footnote{We allow $\varphi$ to return both alternatives in case of tie.} The two agents agree in the already discussed properties, \textsf{Nonemptiness}, \textsf{Anonymity}, \textsf{Unanimity} and \textsf{Continuity}. The last desirable criterion of \textsf{Consistency} depends on the three different contexts: 

\

\noindent \textbf{First context.} Jeremy and John do not share the same belief on the weather during their stay in Oxford. Jeremy believes that the weather will be sunny with probability $p_1$ and raining with probability $1-p_1$ and John thinks that the weather will be sunny with probability $p_2$ and raining with probability $1-p_2$, where $p_1 \neq p_2$. Their expected utilities can be represented by the matrix:
\begin{center}
$ \textbf{p}_{\{1,2\}} \cdot \textbf{u} + \big(I-\textbf{p}_{\{1,2\}}\big) \cdot \textbf{v} = \begin{bmatrix}
 p_1u_1(\ell)+(1-p_1)v_1(\ell) & p_1u_1(c)+(1-p_1)v_1(c)  \\
 p_2u_2(\ell)+(1-p_2)v_2(\ell) & p_2u_2(c)+(1-p_2)v_2(c)
\end{bmatrix}$
\end{center}
 In that case, \textsf{Subjective Expected Consistency} seems to be a desirable property. In line with our example, if Jeremy and John agree to go to the laboratory no matter the weather, then they should also choose the same alternative for any possible beliefs on the weather. Namely, if the alternative $(\ell)$, (resp. $(c)$), is selected when the weather is sunny, that is if $\ell$ (resp. $c$) belongs to $\varphi(\textbf{u})$ and if $\ell$, (resp. $c$), is selected when the weather is rainy, that is if $\ell$ (resp. $c$) belongs to $\varphi(\textbf{v})$ then $\ell$ (resp. $c$) is also selected by $\varphi$ for the expected utility profile $\textbf{p}_{\{1,2\}} \cdot \textbf{u} + (I-\textbf{p}_{\{1,2\}}) \cdot \textbf{v}$. Formally, 
 \begin{align*}\varphi(\textbf{u})\cap \varphi(\textbf{v})\subseteq \varphi\big(\textbf{p}_{\{1,2\}} \cdot \textbf{u} + \big(I-\textbf{p}_{\{1,2\}}\big) \cdot \textbf{v}\big).
 \end{align*}

\

 \noindent \textbf{Second context.} Jeremy and John follow the weather forecasting which indicates that it will be sunny with probability $p$ and raining with probability $1-p$. Their expected utilities can be represented by the matrix:
\begin{center}
$ p \textbf{u} + (1-p) \textbf{v} = \begin{bmatrix}
 pu_1(\ell)+(1-p)v_1(\ell) & pu_1(c)+(1-p)v_1(c)  \\
 pu_2(\ell)+(1-p)v_2(\ell) & pu_2(c)+(1-p)v_2(c)
\end{bmatrix}$
\end{center}
If so, \textsf{Objective Expected Consistency} seems to be a desirable property. That is, if Jeremy and John agree to go to the laboratory no matter the weather, then they should also choose the same alternative for any objective probability on the weather. Notice that this property corresponds to a weaker form of subjective expected utility, since consistency of decisions is required only when $p_1=p_2$. Following the same line of arguments, we obtain: $$\varphi(\textbf{u})\cap \varphi(\textbf{v})\subseteq \varphi(p \textbf{u} + (1-p) \textbf{v}).$$

\

 \noindent \textbf{Third context.} Jeremy and John have no clue on the weather and therefore no probability is assigned to the two possible events. Facing this situation, Jeremy and John adopt a more pessimistic approach leading them to compare their utilities at each alternative and consider the \textit{worst-case scenario profile}. This utility profile can be represented by the matrix:
\begin{center}
$\textbf{u}\wedge\textbf{v} = \begin{bmatrix}
 u_1(\ell) \wedge v_1(\ell) & u_1(c)\wedge v_1(c)  \\
 u_2(\ell) \wedge v_2(\ell) & u_2(c) \wedge v_2(c)
\end{bmatrix}$
\end{center}
where for any real numbers $a$ and $b$, the notation $a \wedge b$ stands for the minimum between $a$ and $b$. 
In that event, \textsf{Minimum Consistency} seems to be a desirable property. That is, if Jeremy and John agree to go to the laboratory independently of the weather, then they should also choose the same alternative in the worst-case scenario profile. In the same way, we end up with: $$\varphi(\textbf{u})\cap \varphi(\textbf{v})\subseteq \varphi(\textbf{u} \wedge \textbf{v}).$$

\

\cite{harsanyi1955cardinal}'s aggregation theorem aims to set up a social utility function which aggregates the ordinal preferences of the agents. It ensures that if agents conform to \cite{vontheory} \textit{expected utility model} (henceforth vNM's EU model), then the social utility function agrees with the vNM's EU model as well and further satisfies a Pareto condition as long as it is affine with respect to individuals' utility functions. This theorem, actively debated, has been the issue of many versions or extensions, mainly dealing with technical aspects of the proof or discussing the domains where the result still remains valid\footnote{We refer the reader to \citet{mongin1994harsanyi} for a survey of the literature on these debates.}. Several attempts have been done in the literature to simplify the proof of this result, as for instance by \citet{Hammond1992} and \citet{mandler2005}. Despite of these efforts, \citet{harsanyi1955cardinal}'s theorem still requires a quite demanding proof being kind of hard to understand. Theorem \ref{maintheorem}, against the backdrop of cardinal preferences, has a similar flavour to this result and it is easier to prove. Rather than a social utility function, we propose to design a choice function; the \textsf{Nonemptiness}, \textsf{Objective Expected Consistency} and \textsf{Continuity} axioms, claimed for our choice function, play an analogous role to the vNM's EU model assumption desired for the social utility function, while our \textsf{Unanimity} assumption is used in place of the Pareto's. Finally, our \textsf{Anonymity} axiom leads, as in \cite{harsanyi1955cardinal}'s result, to a rule that treats all agents equally.\\
On the contrary, if individuals instead of dealing with the vNM EU model, are complied with the \textit{max-min rule} then, following similar arguments and with the intention of designing a choice function, we show that such a way to aggregate should agree with the \textsf{Minimum Consistency} principle. The latter is reflected in Theorem \ref{maintheorem2}, established as well in the context of cardinal preferences. Alternatively stated, Theorem \ref{maintheorem2} ensures that max-min individuals are consistent with a choice function which applies the max-min rule to the worst case scenario utility profile. To the best of our knowledge, this procedure has never been done in the literature with respect to the egalitarian principle.

In this article, we provide three main results, each of which corresponds to one of these three different frames of reference. Proposition \ref{impossibilitytheorem}, demonstrates that it is not possible to reconcile the properties of \textsf{Subjective Expected Consistency}, \textsf{Nonemptiness}, \textsf{Anonymity}, \textsf{Unanimity}, and \textsf{Continuity} simultaneously. Theorem \ref{maintheorem} and Theorem \ref{maintheorem2} overcome this limitation. On the one hand, Theorem \ref{maintheorem}, characterizes the Bentham solution by replacing \textsf{Subjective Expected Consistency} with \textsf{Objective Expected Consistency}. On the other hand, Theorem \ref{maintheorem2}, characterizes the Rawls solution by using \textsf{Minimum Consistency} instead of \textsf{Objective Expected Consistency}. The proof of Theorem \ref{maintheorem} is akin to the one of Theorem \ref{maintheorem2}.

\

\noindent \textbf{Structure of the article:} Section \ref{charnash} deals with formal definitions and notations. In Section \ref{mainresults}, we state and prove our main results. Conclusion and remarks are provided in Section \ref{conclusion}.

\

\section{Preliminaries}\label{charnash}

In this section, we first introduce basic notations and present the definitions of Bentham and Rawls solution concepts. Thereafter, we provide and discuss the axioms which are commonly shared by the two solution concepts.

\subsection{Definitions and notations}

For any $n \geq 2$, let $N=\{1,...,n\}$ be a finite set of agents and let $S$ be a finite set of social alternatives (e.g. states, candidates, action profiles). Given a finite set $J$ and a real number $a_j$ for any $j \in J$, the notation $\bigwedge_{j \in J} a_j$ stands for the minimum value of $\{a_j : j \in J\}$. Each agent $i \in N$ is endowed with a utility function $u_i:S \to \mathcal{D}$, where $\mathcal{D}$ is a nonempty convex\footnote{The convexity assumption is necessary to obtain the proofs of Sections \ref{sectionimposs} and \ref{sectionBentham} but can be relaxed in the proof of  Section \ref{sectionrawls}.} subset of $\mathbb{R}$. We further denote by $\textbf{u}=(u_i)_{i \in N}$ the induced utility profile and by $\mathcal{U}=(\mathbb{\mathcal{D}}^S)^N$ the set of all possible utility profiles. \\

\noindent A \textit{solution concept} on $\mathcal{U}$ is a correspondence $\varphi : \mathcal{U} \rightrightarrows S$, i.e., $\varphi$ assigns to each utility profile $\textbf{u} \in \mathcal{U}$ a subset $\varphi(\textbf{u})$ of the social alternatives set $S$. 
Given $\textbf{u} \in \mathcal{U}$, the \textit{Bentham solution} $B: \mathcal{U}  \rightrightarrows S$ is defined by $$B(\textbf{u})=\argmax{s \in S}{\sum_{i \in N}u_i(s)}$$ 
Moreover, given $\textbf{u}\in \mathcal{U}$, the \textit{Rawls solution} $R: \mathcal{U}  \rightrightarrows S$ is defined by $$R(\textbf{u})=\argmax{s\in S}{\bigwedge_{i\in N}u_i(s)}.$$

\

\subsection{Common axiomatic properties} 

In this section, we present the axioms which are involved in the axiomatization of both solution concepts. The first axiom usually makes part of a social choice function's definition. \\

\noindent \textbf{\textsf{Nonemptiness}}: For any $\textbf{u} \in \mathcal{U}$, we have $\varphi(\textbf{u}) \neq \emptyset$.

\

\noindent We use the notation $\mathfrak{S}(N)$ for the set of permutations of $n$ players. For any $\sigma \in \mathfrak{S}(N)$, let us put $u_i^{\sigma}(s)=u_{\sigma(i)}(s)$. The induced utility profile is denoted by $\textbf{u}_{\sigma}=(u_i^{\sigma})$. The second axiom ensures that all agents are treated equally. \\

\noindent \textbf{\textsf{Anonymity}}: For any $\textbf{u} \in \mathcal{U}$, we have $\varphi(\textbf{u})=\varphi(\textbf{u}_{\sigma})$.

\

\noindent The following axiom is weaker than the standard axiom of Pareto optimality. It states that as long as there exist alternatives that are unanimously preferred, the solution must select one of them. Given $\textbf{u} \in \mathcal{U}$, let $$M(\textbf{u})=\left\{s^{*} \in S : \forall i \in N, \forall s\in S, u_i(s^*)\geq u_i(s)\right\}$$
\\
\noindent  \textbf{\textsf{Unanimity}}: If $M(\textbf{u}) \neq \emptyset$, then $\varphi(\textbf{u}) \subseteq M(\textbf{u})$.\\

\noindent Next axiom ensures that given a utility profile, we can always find sufficiently close to it, another utility profile sharing common solutions with the primal.\\

\noindent  \textbf{\textsf{Continuity}}: For each $a\in S$, the set $\mathcal{U}_a=\{\textbf{u}\, :\, a\in \varphi(\textbf{u}) \}$ is closed.

\

\section{Main results}\label{mainresults}

In this section, we display our main results. First, we provide an impossibility result concluding that there is no solution concept which satisfies \textsf{Nonemptiness}, \textsf{Anonymity}, \textsf{Unanimity} and \textsf{Subjective Expected Consistency} simultaneously. Thenceforth, we axiomatically characterize the Bentham and Rawls solutions by adding a \textsf{Continuity} axiom in both and replacing the last property of consistency with the \textsf{Objective Expected Consistency} and the \textsf{Minimum Consistency} axioms respectively.

\

\subsection{An impossibility result}\label{sectionimposs}
To formally introduce the first variation of the consistency axiom we denote by $\mathcal{P}_N$ the set of diagonal matrices of the following form:
$$
\textbf{p}_N=\begin{bmatrix}
p_{1} &  0  & \ldots & 0\\
0  &  p_{2} & \ldots & 0\\
\vdots & \vdots & \ddots & \vdots\\
0  &   0       &\ldots & p_{n}
\end{bmatrix},$$ 
where $p_i\in (0,1)$ for each $i\in N$.
We denote by $I$ the matrix in $\mathcal{P}_N$ where $p_i=1$ for any $i \in N$. 

\

As illustrated in the introduction, the axiom of \textsf{Subjective Expected Consistency} enforces consistent collective behaviour in situations where agents make use of subjective probabilities represented by a matrix in $\mathcal{P}_N$. We consider two distinct utility profiles $\textbf{u}$ and $\textbf{v}$ such that the sets of recommended alternatives to the agents overlap and assume that each agent $i \in N$ believes that $\textbf{u}$ occurs with probability $p_i$, while $\textbf{v}$ with probability $1-p_i$. Then, \textsf{Subjective Expected Consistency} indicates that, since there is an alternative selected for the agents in both utility profiles, this alternative should be also selected regardless of individual probabilistic beliefs. This principle can be formalized as follows:

\vspace{0.37cm}

\noindent  \textbf{\textsf{Subjective Expected Consistency}}: For all $\textbf{u},\textbf{v} \in \mathcal{U}$, if $\varphi(\textbf{u}) \cap \varphi(\textbf{v}) \neq \emptyset$, then for all $\textbf{p}_{N} \in \mathcal{P}_N$:
$$\varphi(\textbf{u}) \cap \varphi(\textbf{v}) \subseteq \varphi\left(\textbf{p}_{N} \cdot \textbf{u} + \left(I-\textbf{p}_{N}\right) \cdot \textbf{v}\right).\, \footnote{Given $\textbf{w}\in \mathcal{U}$ and $\textbf{q}_N \in \mathcal{P}_N$, the notation $\textbf{q}_N \cdot \textbf{w}$ stands for the standard matrix product of $\textbf{q}_N$ and $\textbf{w}$.}$$

\

\noindent We have now the material to set up Proposition \ref{impossibilitytheorem}.

\begin{Proposition}\label{impossibilitytheorem}
There is no solution concept which satisfies \textsf{Nonemptiness}, \textsf{Anonymity}, \textsf{Unanimity} and \textsf{Subjective Expected Consistency}.
\end{Proposition}

\begin{proof}
Suppose, for the sake of contradiction, that a solution $\varphi$ can be defined that satisfies the axioms of the theorem. To establish our claim of impossibility, it suffices to provide an example of a situation where the presence of such a $\varphi$ concludes to a contradiction. For this purpose, we return to the example, with Jeremy and John, outlined in the introduction and assume that $u_1(\ell)=v_2(\ell)=u_2(c)=v_1(c)=\alpha$ and $u_1(c)=v_2(c)=u_2(\ell)=v_1(\ell)=\beta$. Hence, we have:
$$\textbf{u}=\begin{bmatrix}
   \alpha  & \beta \\
   \beta  & \alpha
\end{bmatrix}\,\,\,\text{and}\,\,\,\,\textbf{v}=\begin{bmatrix}
   \beta  & \alpha \\
   \alpha  & \beta
\end{bmatrix}$$
Notice that $\textbf{v}$ corresponds to the permutation of the agents' utilities.
Then, using \textsf{Anonymity}, we obtain that $\varphi(\textbf{u})=\varphi(\textbf{v})$ and by \textsf{Nonemptiness} it clearly follows: 
\begin{align*}
\varphi(\textbf{u})\cap \varphi(\textbf{v})\neq \emptyset
\end{align*}
Therefore, by \textsf{Subjective Expected Consistency} we get: 
\begin{align*}
\varphi(\textbf{u}) \cap \varphi(\textbf{v}) \subseteq  \varphi\left(\begin{bmatrix}
 p_1\alpha+(1-p_1)\beta & p_1\beta+(1-p_1)\alpha  \\
 p_2\beta+(1-p_2)\alpha & p_2\alpha+(1-p_2)\beta
\end{bmatrix}\right)
\end{align*}
Reversing the roles of $\textbf{u}$ and $\textbf{v}$ and using again \textsf{Subjective Expected Consistency} we further have:
\begin{align*}\label{expect2}
\varphi(\textbf{u}) \cap \varphi(\textbf{v}) \subseteq  \varphi\left(\begin{bmatrix}
 p_1\beta+(1-p_1)\alpha & p_1\alpha +(1-p_1)\beta   \\
 p_2\alpha +(1-p_2)\beta  & p_2\beta +(1-p_2)\alpha 
\end{bmatrix}\right)
\end{align*}
Combining the last three formulas, we obtain:
\begin{equation}\label{finalnonemptiness}
\varphi\left(\begin{bmatrix}
 p_1\alpha+(1-p_1)\beta & p_1\beta+(1-p_1)\alpha  \\
 p_2\beta+(1-p_2)\alpha & p_2\alpha+(1-p_2)\beta
\end{bmatrix}\right) \cap \varphi\left(\begin{bmatrix}
 p_1\beta+(1-p_1)\alpha & p_1\alpha +(1-p_1)\beta   \\
 p_2\alpha +(1-p_2)\beta  & p_2\beta +(1-p_2)\alpha 
\end{bmatrix}\right)\neq \emptyset 
\end{equation}
Let us now assume that $p_1 \in \left(1/2,1\right]$, $p_2 \in [0,1/2)$ and $\alpha \geq \beta$. Then, we get: 
$$
\begin{cases}
p_1\alpha+(1-p_1)\beta \geq p_1\beta+(1-p_1)\alpha\\
p_2\beta+(1-p_2)\alpha \geq p_2\alpha+(1-p_2)\beta
\end{cases}
$$
and it thus follows that in the expected utility profile:
$$
\begin{bmatrix}
 p_1\alpha+(1-p_1)\beta & p_1\beta+(1-p_1)\alpha  \\
 p_2\beta+(1-p_2)\alpha & p_2\alpha+(1-p_2)\beta
\end{bmatrix}
$$
the left column which corresponds to the alternative ($\ell$) of the example is unanimously preferred. On the contrary, in the following expected utility profile:
$$
\begin{bmatrix}
 p_1\beta+(1-p_1)\alpha & p_1\alpha +(1-p_1)\beta   \\
 p_2\alpha +(1-p_2)\beta  & p_2\beta +(1-p_2)\alpha 
\end{bmatrix}
$$
the right column which corresponds to the alternative ($c$) of the example is unanimously preferred. Hence, by \textsf{Unanimity} and \textsf{Nonemptiness}, we get:
$$
\varphi\left(\begin{bmatrix}
 p_1\alpha+(1-p_1)\beta & p_1\beta+(1-p_1)\alpha  \\
 p_2\beta+(1-p_2)\alpha & p_2\alpha+(1-p_2)\beta
\end{bmatrix}\right)=\{\ell\}
$$
and further,
$$
\varphi\left(\begin{bmatrix}
 p_1\beta+(1-p_1)\alpha & p_1\alpha +(1-p_1)\beta   \\
 p_2\alpha +(1-p_2)\beta  & p_2\beta +(1-p_2)\alpha 
\end{bmatrix}\right)=\{c\}
$$
It therefore clearly follows:
$$
\varphi\left(\begin{bmatrix}
 p_1\alpha+(1-p_1)\beta & p_1\beta+(1-p_1)\alpha  \\
 p_2\beta+(1-p_2)\alpha & p_2\alpha+(1-p_2)\beta
\end{bmatrix}\right) \cap \varphi\left(\begin{bmatrix}
 p_1\beta+(1-p_1)\alpha & p_1\alpha +(1-p_1)\beta   \\
 p_2\alpha +(1-p_2)\beta  & p_2\beta +(1-p_2)\alpha 
\end{bmatrix}\right)=\emptyset
$$
which leads to a contradiction with formula (\ref{finalnonemptiness}).
\end{proof}
A significant assumption of Proposition \ref{impossibilitytheorem} is the heterogeneity of the agents' beliefs. Exploring the relevant literature, \cite{hylland1979impossibility} with respect to cardinal preferences and \cite{mongin1995consistent} regarding the ordinal preferences, show that subjective expected utility along with Pareto optimality lead to an impossibility result when agents' beliefs are diverse. Proposition \ref{impossibilitytheorem} points out an incompatibility between a property tightly linked to Pareto Optimality, the \textsf{Unanimity} axiom, and a form of mixture preserving property, the \textsf{Subjective Expected Consistency}. This effect keeps the coherence with other impossibility results, established in the frame of ordinal preferences, which are notably more challenging to prove. \\
Henceforth, we will study the case where $p_1=p_2=...=p_n$, which can be seen as a situation where agents possess an objective probability over the possible events.

\

\subsection{Axiomatic characterization of Bentham solution}\label{sectionBentham}
In this section, we provide an axiomatic characterization of the Bentham solution using \textsf{Nonemptiness}, \textsf{Unanimity}, \textsf{Anonymity}, \textsf{Continuity} axioms and the variation of the consistency axiom called \textsf{Objective Expected Consistency}. The latter property is based on consistent behaviour in related utility profiles with respect to a solution concept. \\
For this purpose, we study two utility profiles such that the sets of recommended alternatives to the agents overlap. Consider a situation where a coin toss decides which of the two utility profiles occurs. Then, \textsf{Objective Expected Consistency} prescribes that, since there is an alternative selected for the agents in both utility profiles, this alternative should also be selected before the coin toss. Namely, given two utility profiles and a probability $p$ over them, the solution set of their expected utility profile under $p$ includes the intersection of their solution sets unless it is empty. Formally,

\

\noindent  \textbf{\textsf{Objective Expected Consistency}}: For all $\textbf{u},\textbf{v} \in \mathcal{U}$, if $\varphi(\textbf{u}) \cap \varphi(\textbf{v}) \neq \emptyset$, then for all $p\in (0,1)$:
$$\varphi(\textbf{u}) \cap \varphi(\textbf{v}) \subseteq \varphi(p\textbf{u} + (1-p)\textbf{v})$$

\vspace{0.43cm}

\noindent Lemma \ref{lemmasum} will be useful in the proof of Theorem \ref{maintheorem}. 
\begin{Lemma}\label{lemmasum}
Assume that $\varphi$ satisfies \textsf{Objective Expected Consistency}. Let $m \geq 2$ and consider $\textbf{u}_1,\textbf{u}_2,...,\textbf{u}_m\in \mathcal{U}$. If there exists $s \in S$ such that $s \in \varphi(\textbf{u}_i)$ for each $i\in \{1,\ldots,m\}$, then $$s \in \varphi\left(\frac{1}{m}\sum_{i\in \{1,\ldots,m\}}\textbf{u}_i\right).$$
\end{Lemma}
\begin{proof}
Given $\textbf{u}_1$, $\textbf{u}_2$,..., $\textbf{u}_m \in \mathcal{U}$, assume that there exists $s \in S$ such that for any $i \in \{1,\ldots,m\}$, $s\in \varphi(\textbf{u}_i)$. We proceed by induction on $m$. If $m=2$, then the result is straightforward by \textsf{Objective Expected Consistency}. We assume that the result is true for $m>2$. It then follows,
\begin{align*}
s \in \varphi\left(\frac{1}{m}\sum_{i\in \{1,\ldots,m\}}\textbf{u}_i\right).
\end{align*}
Let us put $\textbf{w}=\frac{1}{m}\sum_{i\in \{1,\ldots,m\}}\textbf{u}_i$ and consider $\textbf{u}_{m+1} \in \mathcal{U}$, such that $s\in \varphi(\textbf{u}_{m+1})$. \\
Hence, by \textsf{Objective Expected Consistency}, we get
\begin{align*}
s\in \varphi\left(\frac{m}{m+1}\textbf{w}+\frac{1}{m+1}\textbf{u}_{m+1}\right)=\varphi\left(\frac{1}{m+1}\sum_{i\in \{1,\ldots,m+1\}}\textbf{u}_i\right)
\end{align*} 
which concludes the proof.
\end{proof}

\

\noindent In his seminal paper, \cite{harsanyi1955cardinal} proposes an axiomatic foundation of utilitarianism wherein a criterion closely related to \textsf{Objective Expected Consistency} together with Pareto optimality imply social utility to be a combination of individual utilities.
\\
Thereafter, we provide Theorem \ref{maintheorem}, which consists an axiomatic characterization of the Bentham solution. It certainly has a flavor of the \citet{harsanyi1955cardinal}'s result on the aggregation of ordinal preferences the proof of which, as previously mentioned, is quite tough. Our proof remains simpler within the context of cardinal preferences.

\begin{Theorem}\label{maintheorem}
A solution concept satisfies \textsf{Nonemptiness}, \textsf{Anonymity}, \textsf{Unanimity}, \textsf{Continuity} and \textsf{Objective Expected Consistency} if and only if it is the Bentham solution.
\end{Theorem}

\begin{proof}
We prove only the sufficiency of the axioms. We proceed in two steps.

\

\noindent{\textbf{Step 1.}} \textit{If $\varphi$ satisfies \textsf{Anonymity}, \textsf{Unanimity} and \textsf{Objective Expected Consistency}, then we have $\varphi(\textbf{u})\subseteq B(\textbf{u})$}.

\

\noindent Let $a$ belong to $\varphi\left(\textbf{u}\right)$ and let $\sigma \in \mathfrak{S}(N)$ be the cyclic permutation on $N$ defined as: $\sigma(1)=2, \sigma(2)=3,\ldots, \sigma(n-1)=n$ and $\sigma(n)=1$. 
Let $\sigma^k$ represent the permutation obtained by composing $\sigma$ $k$-times, such that $\sigma^k(r)=k+r$ if $k+r\leq n$ and $\sigma^k(r)=n-k+r$ if $k+r>n$. By \textsf{Anonymity}, for each $k\in \{1,\ldots,n\}$, $a\in \varphi(\textbf{u}_{\sigma^k})$. Since by assumption $\varphi$ verifies \textsf{Objective Expected Consistency}, Lemma \ref{lemmasum} with $m=n$ implies that,
\begin{align*}
a\in \varphi\left(\frac{1}{n}\sum_{k=1}^n\textbf{u}_{\sigma^k}\right).
\end{align*}
Let $\textbf{v}$ be the utility profile defined by $\textbf{v}=\sum_{k=1}^n\textbf{u}_{\sigma^k}$. It is easy to observe that for each alternative $s \in S$ and each individual $i \in N$, 
$$v_i(s)=\sum_{j\in N}u_j(s).$$
Therefore, for any $i,j \in N$ and any $s \in S$ we have $v_i(s)=v_j(s)$.
By Unanimity, it follows that $$\varphi\left(\frac{1}{n}\sum_{k=1}^n\textbf{u}_{\sigma^k}\right)\subseteq B(\textbf{u}),$$
and we thus get $a\in B(\textbf{u})$.

\

\noindent{\textbf{Step 2.}} \textit{If $\varphi$ satisfies \textsf{Nonemptiness}, \textsf{Anonymity}, \textsf{Unanimity}, \textsf{Continuity} and \textsf{Objective Expected Consistency}, then we have $B(\textbf{u}) \subseteq \varphi(\textbf{u}) $}.

\noindent Let $a\in B(\textbf{u})$ and $(\textbf{w}^n)_{n\in \mathbb{N}^*}$ be a sequence of utility profiles defined for each $i\in N$ and $n \in \mathbb{N}^*$, by $w_i(a)=u_i(a)+\frac{1}{n}$ and for any $s\neq a$, by $w_i(s)=u_i(s)$. It is clear that for any $n\in \mathbb{N}^*$, $B(\textbf{w}^n)=\{a\}$. The axiom of \textsf{Nonemptiness} and the Step 1 imply that for each $n\in \mathbb{N}^*$, $\varphi(\textbf{w}^n)=\{a\}$. Since $(\textbf{w}^n)_{n\in \mathbb{N}^*}$ converges to $\textbf{u}$ when $n$ goes to infinity, the axiom of \textsf{Continuity} implies that $a\in \varphi(\textbf{u})$, as desired.

\end{proof}

\
\begin{Remark}
    It is worth noting that Proposition \ref{impossibilitytheorem} can be deduced from the following observations:
    \begin{enumerate}
        \item From the first step of the proof, any solution that satisfies \textsf{Anonymity}, \textsf{Unanimity}, and \textsf{Objective Expected Consistency} is a subsolution of the Bentham solution.
        \item A solution that satisfies \textsf{Subjective Expected Consistency} will also satisfy \textsf{Objective Expected Consistency}.
        \item However, the Bentham solution does not satisfy \textsf{Subjective Expected Consistency}.
    \end{enumerate}
    To bolster the last point, we can reference a counterexample similar to the one used in establishing Proposition \ref{impossibilitytheorem}.
\end{Remark}

\subsection{Axiomatic characterization of Rawls solution}\label{sectionrawls}
In this section, we contribute an axiomatic characterization of the Rawls solution using \textsf{Nonemptiness}, \textsf{Unanimity}, \textsf{Anonymity}, \textsf{Continuity} axioms and the variation of the consistency axiom named \textsf{Minimum Consistency}. Accordingly, we identify a single axiom which distinguishes Rawls and Bentham solutions. \\
In doing so, we introduce an alternative formulation of the consistency axiom that is more appropriate in situations of complete uncertainty. In such a case, there is no information regarding the probability whereby a utility profile occurs (either $\textbf{u}$ or $\textbf{v}$) and as a consequence, it is not possible to use expected utility. A natural and widely used in the literature procedure is to adopt a pessimistic approach and assume the worst possible outcome concerning the agents' utilities; namely, the worst-case scenario profile, represented by $\mathbf{u\wedge v}$, is defined for any player $i \in N$ and any alternative $s \in S$ by
\begin{align*}
[\mathbf{u\wedge v}]_i(s)=u_i(s)\wedge v_i(s)
\end{align*}

The following criterion of collective decision in uncertain environment states that as long as there exists an alternative selected by the society in both utility profiles then this alternative should also be selected by the society in the worst-case scenario profile $\mathbf{u\wedge v}$. This principle can be written formally as follows:

\vspace{0.37cm}

\noindent  \textbf{\textsf{Minimum Consistency}}: For all $\textbf{u},\textbf{v} \in \mathcal{U}$, if $\varphi(\textbf{u}) \cap \varphi(\textbf{v}) \neq \emptyset$, then:
$$\varphi(\textbf{u}) \cap \varphi(\textbf{v}) \subseteq \varphi(\mathbf{u\wedge v})$$

\

\begin{Theorem}\label{maintheorem2}
A solution concept satisfies \textsf{Nonemptiness}, \textsf{Anonymity}, \textsf{Unanimity}, \textsf{Continuity} and \textsf{Minimum Consistency} if and only if it is the Rawls solution.
\end{Theorem}
\begin{proof}
           
We prove only the sufficiency of the axioms. We proceed in two steps. 

\

\noindent{\textbf{Step 1.}} \textit{If $\varphi$ satisfies \textsf{Nonemptiness}, \textsf{Anonymity}, \textsf{Unanimity} and \textsf{Minimum Consistency}, then we have $\varphi(\textbf{u})\subseteq R(\textbf{u})$}.

\

\noindent Let $a$ belong to $\varphi\left(\textbf{u}\right)$ and let $\sigma \in \mathfrak{S}(N)$ be the cyclic permutation on $N$ defined as in the Step 1 of Theorem \ref{maintheorem}. By \textsf{Anonymity}, we get that for each $k\in \{1,\ldots,n\}$, $a\in \varphi(\textbf{u}_{\sigma^k})$. Since by assumption $\varphi$ verifies \textsf{Minimum Consistency}, a simple induction on $n$ implies that, 
\begin{align*}
a\in \varphi\left(\bigwedge_{k=1}^n\textbf{u}_{\sigma^k}\right).
\end{align*}
Let $\textbf{v}$ be the utility profile defined by $\textbf{v}=\bigwedge_{k=1}^n\textbf{u}_{\sigma^k}$. It is easy to see that for each alternative $s \in S$ and each individual $i \in N$, we have
\begin{align*}
v_i(s)=\bigwedge_{j\in N}u_j(s).
\end{align*}
Therefore, for any $i,j \in N$ and for any $s \in S$ we have $v_i(s)=v_j(s)$.
Then, by Unanimity,
\begin{align*}
\varphi\left(\bigwedge_{k=1}^n\textbf{u}_{\sigma^k}\right)\subseteq R(\textbf{u}),
\end{align*}
and we thus get $a\in R(\textbf{u})$.

\

\noindent{\textbf{Step 2.}} \textit{If $\varphi$ satisfies \textsf{Nonemptiness}, \textsf{Anonymity}, \textsf{Unanimity}, \textsf{Minimum Consistency} and \textsf{Continuity} then we have $R(\textbf{u}) \subseteq \varphi(\textbf{u})$}.

\

\noindent Let $a \in R(\textbf{u})$ and $(\textbf{w}^n)_{n\in \mathbb{N}^*}$ be a sequence of utility profiles defined for each $i\in N$ and $n \in \mathbb{N}^*$, by $w_i(a)=u_i(a)+\frac{1}{n}$ and for each $s\neq a$, by $w_i(s)=u_i(s)$. It is clear that for any $n\in \mathbb{N}^*$, $R(\textbf{w}^n)=\{a\}$. The axiom of \textsf{Nonemptiness} and the Step 1 imply that for each $n\in \mathbb{N}^*$, $\varphi(\textbf{w}^n)=\{a\}$. Using the axiom of \textsf{Continuity}, since $(\textbf{w}^n)_{n\in \mathbb{N}^*}$ converges to $\textbf{u}$ when $n$ goes to infinity, we conclude that $a\in \varphi(\textbf{u})$, as desired.
\end{proof}

\

The proofs of Theorem \ref{maintheorem} and Theorem \ref{maintheorem2} follow quite similar arguments and steps. It is easy to check that both proofs can accommodate a weaken version of the \textsf{Unanimity} axiom without affecting their validity. This weaken version can be achieved as follows: if all individuals share the same utility function, represented by $\textbf{u}=(u_i)_{i\in N}$, i.e., $u_i=u_j$ for all $i,j\in N$, then $\varphi(\textbf{u}) \subseteq M(\textbf{u})$. This alternative axiom, akin to the notion of faithfulness proposed by \cite{young1974axiomatization} for the Borda choice correspondence, is clearly implied by \textsf{Unanimity} axiom.

\

\subsection{Independence of the axioms}

In this section, we show the independence of the axioms involved in Theorem \ref{maintheorem} and Theorem \ref{maintheorem2} (the formal proof is left to the reader). Accordingly, we need to define some further solution concepts. We first recall that the unanimous solution $M$ is defined for each $\textbf{u}\in \mathcal{U}$ by $$M(\textbf{u})=\left\{s^{*} \in S : \forall i \in N, \forall s\in S, u_i(s^*)\geq u_i(s)\right\}.$$ Let $i\in N$ be an arbitrary agent, dubbed the dictator. We say that $dict^i: \mathcal{U} \rightrightarrows S$ is the \textit{dictatorial solution} with respect to $i$, if for each $\textbf{u}\in \mathcal{U}$, $$dict^i(\textbf{u})=\argmax{s\in S}{u_i(s)}.$$ We say that $AB: \mathcal{U} \rightrightarrows S$ is the \textit{anti-Bentham solution}, if for each $\textbf{u}\in \mathcal{U}$, $$AB(\textbf{u})=B(-\textbf{u}).$$ 
In addition, we define $\widetilde{R}:\mathcal{U}\rightrightarrows S $ to be the $\prec$-\textit{subRawls} solution, where $\prec$ is a linear order of $S$. Let us denote $\min_{\prec}X$ the minimal elements of $X\subseteq S$ with respect to $\prec$. $\widetilde{R}$ is defined for each $u\in \mathcal{U}$ by $\widetilde{R}(\textbf{u})=\min_{\prec}R(\textbf{u})$.\\
Furthermore, we define $\widetilde{B}:\mathcal{U}\rightrightarrows S $ to be the $\prec$-\textit{subBentham} solution, where $\prec$ is a linear order of $S$. $\widetilde{B}$ is defined for each $u\in \mathcal{U}$ by $\widetilde{B}(\textbf{u})=\min_{\prec}B(\textbf{u})$.
The following table illustrates the independence of the axioms appeared in Theorem \ref{maintheorem} and Theorem \ref{maintheorem2}. Indeed, regarding each axiomatization, we can always find another solution concept that does not satisfy simultaneously all the four axioms. 
 \medskip 
\begin{center}
\begin{tabular}{l|c|c|c|c|c|c|} 

          & NE & A  & U & C & OEC & MC \\ 
  \hline
 $M$      & no & yes & yes & yes & yes & yes \\ 
 $dict^i$ & yes & no & yes & yes & yes & yes \\ 
 $AB$     & yes & yes & no & yes & yes & no \\ 
 $\widetilde{R}$  & yes & yes & yes & no & no  & yes \\
 $\widetilde{B}$  & yes & yes & yes & no & yes  & no \\
 $R$      & yes & yes & yes & yes & no  & yes  \\ 
 $B$      & yes & yes & yes & yes & yes &no \\
 \hline
\end{tabular}
\end{center}

\

\section{Conclusion}\label{conclusion}

\hspace{0.4cm} In this article, three variations of the widely known consistency property are introduced. Using the former called \textsf{Subjective Consistency}, we obtain an impossibility result, using the second named \textsf{Objective Consistency}, we establish an axiomatization of classical utilitarianism and using the latter referred to as \textsf{Minimum Consistency}, we characterize the egalitarian principle.
\\
It is further revealed that although the two solutions are philosophically in conflict, their axiomatizations are closely related; that is to say, all four of the following properties: \textsf{Nonemptiness}, \textsf{Unanimity}, \textsf{Anonimity}, and \textsf{Continuity} are commonly shared by the two solution concepts and moreover the proofs of their axiomatic foundations follow a quite similar reasoning. As a result, the two characterizations disagree only with the choice of the appropriate, with respect to each concept, variation of consistency. On the one hand, utilitarianism calls for the \textsf{Objective Consistency} property and therefore, potential connections with the cornerstone work of \cite{harsanyi1955cardinal} can been achieved provided that individuals conform to the vNM's EU model. On the other hand, egalitarianism deals with the \textsf{Minimum Consistency} axiom and thus an interpretation of how a choice function should aggregate cardinal preferences as long as individuals comply with the max-min criterion, is now possible.

\

\

\section*{Acknowledgments}
This research was funded by the French National Research Agency (ANR) under the Citizens project "ANR-22-CE26-0019-01". Financial support of MODMAD is also gratefully acknowledged. The authors also want to thank Marco Li Calzi, Eric R\'emila and Vassili Vergopoulos for valuable comments.

\bibliographystyle{plainnat}
\bibliography{biblio}
\end{document}